\documentclass[11pt]{article}

\usepackage{amsthm}
\usepackage{amsmath,amssymb,amsfonts}
\usepackage{mathtools}
\usepackage[a4paper, total={160mm,247mm},
 left=25mm,
 top=25mm]{geometry}

\usepackage{authblk}

\theoremstyle{plain}
\newtheorem{theorem}{Theorem}
\newtheorem{proposition}[theorem]{Proposition}
\newtheorem{corollary}{Corollary}

\theoremstyle{remark}
\newtheorem{example}{Example}
\newtheorem{remark}{Remark}

\theoremstyle{definition}

\newtheorem{construction}{Construction}

\usepackage[font=small,width=.85\textwidth]{caption}

\begin{document}

\title{Codes with restricted overlaps: \\ expandability, constructions, and bounds}

\author{Lidija Stanovnik}
\date{}

\affil{University of Ljubljana, Faculty of Computer and Information Science,\\ Večna pot 113, Ljubljana, Slovenia}

\maketitle

%%==================================%%
%% Sample for unstructured abstract %%
%%==================================%%

\abstract{\noindent Consider a $q$-ary block code satisfying the property that no $l$-letters long codeword's prefix occurs as a suffix of any codeword for $l$ inside some interval. We determine a general upper bound on the maximum size of these codes and a tighter bound for codes where overlaps with lengths not exceeding $k$ are prohibited. We then provide constructions for codes with various restrictions on overlap lengths and use them to determine lower bounds on the maximum sizes. In particular, we construct $(1,k)$-overlap-free codes where $k \geq n/2$ and $n$ denotes the block size, expand a known construction of $(k,n-1)$-overlap-free codes, and combine the ideas behind both constructions to obtain $(t_1,t_2)$-overlap-free codes and codes that are simultaneously $(1,k)$- and $(n-k,n-1)$-overlap-free for some $k < n/2$. In the case when overlaps of lengths between 1 and $k$ are prohibited, we complete the characterisation of non-expandable codes started by Cai, Wang, and Feng (2023)~\cite{cai:2023}.
}

%\keywords{$(1,k)$-overlap-free code, $(t_1,t_2)$-overlap-free code, non-overlapping code, weakly mutually uncorrelated code}

\section{Introduction}
Two not necessarily distinct words $u$ and $v$ over an alphabet $\Sigma$ have a $t$-overlap if the $t$ letters long prefix of $u$ occurs as a suffix of $v$, or vice-versa. A code $C \subseteq \Sigma^n$ is $(t_1,t_2)$-overlap-free if for all $t_1 \leq t \leq t_2$ no codewords in $C$ have a $t$-overlap~\cite{blackburn:2023}.
These codes may be useful in channels where bursts of at least $\min\; \{t_1,n-t_2\}$ and at most $\min\; \{n-t_1,t_2\}$ deletions or insertions are expected. They are guaranteed to detect deletions after $2n$ symbols are read and insertions after $3n$ symbols are read. Some parameter values may even perform better.

So far, only special cases of $(t_1,t_2)$-overlap-free codes have been studied. Non-overlapping codes, i.e., the $(1,n-1)$-overlap-free codes, were established for synchronisation sixty years ago~\cite{levenshtein:1970}. 
At the beginning of this millennium, $T$-shift synchronisation codes were considered for data storage~\cite{ahlswede:2005,ahlswede:2008}. Every uni- and bidirectional $T$-shift synchronisation code is a $(n-T,n-1)$-overlap-free code, but the two classes only coincide when $T=1$. 
In 2015, the term $\kappa$-weakly mutually uncorrelated codes was introduced for $(\kappa, n-1)$-overlap-free codes that were motivated by applications in DNA-based data storage~\cite{Yazdi:2015,Yazdi:2018,Kumar:2023,Lu:2023,Chee:2020}.
Recently, Blackburn et al., who defined the $(t_1,t_2)$-overlap-free codes, spotlighted the class of $(1,k)$-overlap-free codes that have an advantage over the class of $(n-k,n-1)$-overlap-free codes in channels where deletions are more common than insertions~\cite{blackburn:2023}. They also suggested that codes that are simultaneously $(1,k)$-overlap-free and $(n-k,n-1)$-overlap-free for some $k < n/2$ might be fruitful for applications since they inherit their deletion-detection capabilities from $(1,k)$-overlap-free codes and insertion-detection capabilities from $(n-k,n-1)$-overlap-free codes. In the same year, Cai, Wang, and Feng~\cite{cai:2023} showed that the study of a $(1,k)$-overlap-free code of length $n$ for $n \geq 2k$ can be reduced to the study of a $(1,k)$-overlap-free code of length $2k$. This is because for any $n > 2k$, the $(1,k)$-overlap property is independent of the symbols placed at positions $k+1, \dots, n-k$. The $(1,k)$-overlap-free codes where $n/2 < k < n-1$ have not been studied much. To our knowledge, only one construction of $(1,n-2)$-overlap-free codes exists. It was presented in \cite{cai:2023}.

A code that attains the largest possible size is called \emph{maximum}.
 Herein, the maximum size of a $(t_1,t_2)$-overlap-free code of length $n$ over an alphabet with $q$ elements will be denoted by $S_{t_1}^{t_2}(q,n)$.
The tightest known upper bound on the size of maximum non-overlapping codes given in Eq.~\eqref{eq:1} was determined in \cite{levenshtein:1970} by studying the corresponding generating function. In the construction from \cite{Blackburn:2015}, the upper bound is attained when $n$ divides $q$. 
\begin{align}\label{eq:1}
    S_{1}^{n-1}(q,n) \leq \left( \frac{n-1}{n} \right)^{n-1} \frac{q^n}{n}.
\end{align}
The first upper bound on $S_k^{n-1}(q,n)$ was determined in \cite{Yazdi:2018}. Its improved version published in \cite{blackburn:2023} is given in Eq.~\eqref{eq:2}.
\begin{align}
    S_k^{n-1}(q,n) &\leq \frac{q^n}{2n-2k+1}. \label{eq:2}
\end{align}
A similar result on $S_1^k(q,n)$ for $k \leq n/2$ was first derived in \cite{blackburn:2023}.
The upper bound was improved in \cite{cai:2023} (Eq.~\eqref{eq:3} below). Moreover, in the same article, two exact values were determined. They are given in Eqs.~\eqref{eq:4} and \eqref{eq:5} below. In Eq.~\eqref{eq:5}, $\lfloor x\rceil$ denotes the rounding of $x$ to the nearest integer.
\begin{align}
    S_1^k(q,n) &\leq \frac{q^n - q^{n-k}}{2k}, \; k \leq n/2, \label{eq:3}\\
    S_1^k(q,n) &= q^{n-2k}S_1^k(q,2k), \; k \leq n/2, \label{eq:4}\\
    S_1^2(q,n) &= \Bigl\lfloor\frac{q}{3}\Bigr\rceil \left(q - \Bigl\lfloor\frac{q}{3}\Bigr\rceil\right)^2 q^{n-3}, \; n \geq 4. \label{eq:5}
\end{align}
The author is not aware of any existing bounds for other values of $t_1$ and $t_2$. Hence, this paper aims to derive new upper and lower bounds on $S_{t_1}^{t_2}$ and provide corresponding constructions.

The outline of the paper is as follows. Section \ref{section:preliminaries} introduces the necessary notation and summarises existing results. A general upper bound on $S_{t_1}^{t_2}$ and a tighter upper bound for the size of $(1,k)$-overlap-free codes when $k < n \leq 2k$ is determined in Section \ref{section:upper}. In Section~\ref{section:constructions}, we develop methods to construct both types of codes and expand a construction of $(k, n-1)$-overlap-free codes from \cite{Yazdi:2018}. Moreover, we characterise maximal $(1,k)$-overlap-free codes and construct codes that are simultaneously $(1,k)-$ and $(n-k,n-1)$-overlap-free.
Finally, these constructions are used to determine the corresponding lower bounds on $S_{t_1}^{t_2}$ in Section~\ref{section:lower}.

\section{Preliminaries}\label{section:preliminaries}
 This paper will use the following notation. The set of integers $1,\dots, n$ is denoted by $\left[n\right]$, and the set of integers from $i$ to $n$ by $\left[i,n\right]$.
 For sets of words, $L$ and $R$, $LR$ denotes the concatenation of sets $L$ and $R$, i.e. the set of all words of the form $lr$, where $l \in L$ and $r \in R$. If $x$ is a word, we write $xR$ and $Lx$ to denote the sets $\{x\}R$ and $L\{x\}$ respectively. $L^i$ denotes the concatenation of the set $L$ with itself $i$ times. The Kleene star denotes the smallest superset that is closed under concatenation and includes the empty set $L^* = \bigcup_{i \geq 0} L^i$, and the Kleene plus denotes the smallest superset that is closed under concatenation and does not include the empty set  $L^+ = \bigcup_{i > 0} L^i$.

We consider words over an alphabet $\Sigma$ with at least two letters and denote the number of letters in $\Sigma$ by $q$. Whenever we deal with a code $C \subseteq \Sigma^n$, we assume $n \geq 2$. The length of the word $w$ is denoted by $\lvert w \rvert$.
A word $w \in \Sigma^+$ has a \emph{period} $i$ if $w = (w_1\cdots w_i)^+$. The smallest number $i$ satisfying this relation is called the \emph{least period}, and words with the least period $\lvert w \rvert$ are called \emph{primitive}. It is known (see e.g. \cite{petersen:1996}) that the number of primitive words in $\Sigma^n$ equals
 \begin{align}
     P_q(n) = \sum_{d\mid n} \mu(d) q^{n/d},
 \end{align}
where $\mu(d)$ is the Möbius function
\begin{align*}
    \mu(n) = \begin{cases}
        (-1)^t & \text{if } n = p_1p_2\cdots p_t \text{ for distinct primes } p_i, \\
        1 & \text{if } n = 1, \\
        0 & \text{if } n \text{ is divisible by a square.}
    \end{cases}
\end{align*}
A \emph{composition} of a positive integer $n$ with $k$ parts is a $k$-tuple of positive integers whose sum equals $n$. We denote the number of parts of composition $\alpha$ by $par(\alpha)$, and its size by $\lvert \alpha \rvert \coloneqq \alpha_1 + \cdots + \alpha_{par(\alpha)}$.

Let $(L_i,R_i)_{i < k}$ be a family of pairs of disjoint sets that satisfy the following constraints. The sets $L_1$ and $R_1$ are non-empty, and their union equals $\Sigma$. For every $i \in \left[2,k\right]$, the sets $L_i$ and $R_i$ are allowed to be empty, but their union should equal $\bigcup_{j\in\left[i-1\right]} L_j R_{i-j} $. The authors of \cite{fimmel:2023} proposed that such a family of pairs of sets can be used to construct non-overlapping codes (see Construction~\ref{construction:fimmel} below), and in particular, that it constructs all \emph{maximal} non-overlapping codes. Recall that a code is called maximal if the inclusion of any new word violates the code's definition. Otherwise, it is called \emph{expandable}. 

\begin{construction}
    \label{construction:fimmel}   \cite{fimmel:2023,stanovnik:2024}
    The code $C = \bigcup_{i\in\left[n-1\right]} L_i R_{n-i} \subseteq \Sigma^n$ is non-overlapping.
\end{construction}

In \cite{stanovnik:2024}, we have shown that Construction~\ref{construction:fimmel} provides all maximal non-overlapping and some expandable non-overlapping codes. We will build new constructions of $(t_1,t_2)$-overlap-free codes upon it. To prove their correctness, we need some properties of the family of pairs of disjoint sets $(L_i,R_i)_{i <k}$. They are given in Propositions~\ref{proposition:z} and \ref{proposition:unique_index} (below). 

\begin{proposition}
  \label{proposition:z}
  \cite{stanovnik:2024}
    Let $C=\bigcup_{i\in\left[n-1\right]} L_i R_{n-i}$ be a code from Construction~\ref{construction:fimmel} and $X \coloneqq \bigcup_{i\in\left[2n+1\right]} X_i$, such that $X_{2i-1} = L_i$, $X_{2i} = R_i$, and $X_{2n+1} = C$. Then no proper prefix in $X$ occurs as a proper suffix in $X$.
\end{proposition}

\begin{proposition}\label{proposition:unique_index}
  \cite{stanovnik:2024}
  Let $C = \bigcup_{i\in\left[n-1\right]} L_i R_{n-i}$ be a code obtained from Construction~\ref{construction:fimmel}.
  For every word $w \in C$, there exists a unique index $i$ such that $w \in L_iR_{n-i}$.
\end{proposition}

Both claims were proven in \cite{stanovnik:2024} using a sequence of decompositions of words that was initially defined by Fimmel et al. in \cite{fimmel:2019}. We will use a generalisation of this approach in this article. For a family of pairs of sets $(L_i,R_i)_{i<k}$ that satisfy the same requirements as before, define $L \coloneqq \bigcup_{i < k} L_i$ and $R \coloneqq \bigcup_{i < k} R_i$.
For $w \in L_1 (L_1 \cup R_1)^* R_1$ 
define $(p_m(w))_{m \geq 0}$ to be a sequence of decompositions of $w$ in $(L \cup R)^+$ encoded with a binary word $p_m$ over the alphabet $\{\text{l},\text{r}\}$ such that
        $p_0(w) \in \text{l}\{\text{l},\text{r}\}^{n-2}r$ with $p_0(w)_i = \text{l}$ if $w_i \in L_1$ and $p_0(w)_i = \text{r}$ if $w_i \in R_1$.
        If $p_m(w)$ contains a subword lr that corresponds to a subword of $w$ of length at most $k$, then $p_{m+1}(w) \in \{\text{l},\text{r}\}^+$ is obtained by replacing each occurrence of lr in $p_m(w)$ by l if lr corresponds to a subword $w'$ of $w$ in $L_iR_j\subseteq L_{i+j}$ for some positive integers $i$ and $j$ such that $i+j \leq k$, or by r if it corresponds to a subword $w'$ of $w$ in $L_iR_j \subseteq R_{i+j}$ for some positive integers $i$ and $j$ such that $i+j\leq k$. Occurrences of lr in $p_m(w)$ that correspond to subwords longer than $k$ are not modified.
The sequence $(p_m(w))_{m\geq 0}$ is finite, since, whenever defined, the decomposition $p_{m+1}(w)$ is strictly shorter than decomposition $p_m(w)$. Below, we provide an example to illustrate how this sequence of decompositions is performed.

\begin{example}
    Let $k=2$, $\Sigma = \{0,1,2\}$, and $(L_1,R_1) = (\{0,1\},\{2\})$, $(L_2,R_2) = (\{02\}, \{12\})$. Verifying that the family of pairs of sets satisfies the required conditions is easy. We will perform the sequence of decompositions $(p_m(w))_{m \geq 0}$ on the word $02122$. It is a word over $\Sigma$, it starts with a letter from $L_1$ and ends with a letter from $R_1$, so the first element of the sequence $p_0(02122)$ is defined. We obtain it by replacing $0$'s and $1$'s with l (they belong to $L_1$) and $2$'s with r (it belongs to $R_1$). Hence, $p_0(02122) = \text{lrlrr}$.
    To compute $p_1(02122)$, we first observe that there are two lr's in $p_0(02122)$. The first one corresponds to the subword 02, and the second one corresponds to the subword 12. The subword 02 is an element of $L_2$, so we replace it with l. The subword 12 is an element of $R_2$; thus, it is replaced with r. Therefore, $p_1(02122) = \text{lrrr}$. There is one lr is $p_1(02122)$. It corresponds to the subword 0212, which is longer than $k$. Therefore, $p_2(02122)$ is not defined.
\end{example}

\section{Upper bounds}\label{section:upper}
In this section, we determine new upper bounds on the size of $(t_1,t_2)$-overlap-free codes.
We first observe that the arguments from \cite{cai:2023} on maximum $(1,k)$-overlap-free codes for $k < n/2$ hold also for $(t_1,t_2)$-overlap-free codes with $t_2 < n/2$. In particular, every $(t_1,t_2)$-overlap-free code $C$ of length $2t_2$ can be extended into a $(t_1,t_2)$-overlap-free code $X = \{x_1\cdots x_{t_2} \Sigma^{n-2t_2} x_{t_2+1}\cdots x_{2t_2} \mid x_1\cdots x_{2t_2} \in C\}$ of length $n > 2t_2$, and removing the letters at positions $t_2+1$ to $n-t_2$ from every word in a $(t_1,t_2)$-overlap-free code of length $n>2t_2$ yields a $(t_1,t_2)$-overlap-free code of length $2t_2$. Hence, the following equality holds.

\begin{proposition}\label{proposition:generalise_cai}
    If $n > 2t_2$,
    \begin{align*}
        S_{t_1}^{t_2}(q,n) = q^{n-2t_2} S_{t_1}^{t_2}(q,2t_2).
    \end{align*}
\end{proposition}
Now, let us determine an upper bound for the maximum size of the $(t_1,t_2)$-overlap-free codes satisfying $t_2 \geq n/2$. The idea behind its proof comes from the upper bound on the size of non-overlapping codes given in \cite{Blackburn:2015}.

\begin{proposition}
    \begin{align*}
        S_{t_1}^{t_2}(q,n) < \frac{q^n}{\Bigl\lfloor \frac{2n-t_1}{n-t_2} \Bigr\rfloor}.
    \end{align*}
\end{proposition}
\begin{proof}
    Let a code $C$ be $(t_1,t_2)$-overlap-free. Construct a set $X$ consisting of all the words $x_1 \cdots x_{2n-t_1}$ from $\Sigma^{2n-t_1}$ for which there exists a non-negative integer $k$ such that \\ $x_{k(n-t_2)+1}\cdots x_{k(n-t_2)+n}$ is a word in $C$. 
    We will show that for every word in $X$, there is exactly one such $k$. Since there are $\Bigl\lfloor \frac{2n-t_1}{n-t_2} \Bigr\rfloor$ possible values of $k$, this will imply 
    \begin{align*}
         \lvert X \rvert = \lvert C \rvert \biggl\lfloor \frac{2n-t_1}{n-t_2} \biggr\rfloor q^{n-t_1}.
    \end{align*}
    Assume that for some word $x_1\cdots x_{2n-t_1} \in X$ there exist two (not necessarily distinct) words $u$ and $v$ from $C$ that start at distinct positions of $x$. Denote these positions with $k_1(n-t_2)+1$ and $k_2(n-t_2)+1$. Without loss of generality, we can assume $k_1 < k_2$. Since the length of $x$ is shorter than $2n$, an overlap exists between $u$ and $v$. Its length equals $l= n - (k_2-k_1)(n-t_2)$.
    We need to explain that $t_1 \leq l \leq t_2$, contradicting the fact that $C$ is $(t_1,t_2)$-overlap-free.
    If $l$ was larger than $t_2$, it would also hold that
    \begin{align*}
        n-(k_2-k_1)(n-t_2) &> t_2\\
        t_2 (k_2-k_1-1) &> n (k_2-k_1-1) \\
        t_2 &> n,
    \end{align*}
    which, by definition, is not true. If, on the other hand, $l$ was shorter than $t_1$, the length of $x_{k_1 t_1 + n + 1} \cdots x_{2n-t_1}$ should be longer than $n-t_1$, which again cannot hold. To complete the proof, we only need an upper bound on the size of $X$. No constant word can be in $X$, so we simply take $\lvert X \rvert \leq q^{2n-t_1} - q < q^{2n-t_1}$, and the proposition follows.
\end{proof}

Providing a tighter bound for $(1,k)$-overlap-free codes for some $k \geq n/2$ is easy. In particular, any $(1,k)$-overlap-free code is also $(1,\lfloor n/2\rfloor)$-overlap-free, so
\begin{align*}
    S_1^k(q,n) \leq S_1^{\lfloor n/2 \rfloor}(q,n) \leq \frac{q^n - q^{\lceil n/2\rceil}}{2\lfloor n/2 \rfloor},
\end{align*}
by Equation~\eqref{eq:3}. We provide a tighter upper bound.

\begin{proposition}
    If $k \geq n/2$, 
    $S_1^k(q,n) \leq \sum_{d \mid n} \frac{\mu(d)}{n} q^{n/d}$.
\end{proposition}

\begin{proof}
    As explained in \cite{cai:2023}, if $C$ is a $(1,k)$-overlap-free code for some $k \geq n/2$, all words and cyclic shifts of words in $C$ are primitive. More specifically, if a cyclic shift of $w \in C$ had a period $d < n$, i.e., $w_{i+1}\cdots w_n w_1\cdots w_i = (v_1 \cdots v_d)^{n/d}$, there would exist $m \in \left[d\right]$ such that $w_1 = v_m$. Then 
    \begin{align*}
        w_1 \cdots w_{d} = v_m \cdots v_d v_1 \cdots v_{m-1} = w_{n-d+1}\cdots w_n,
    \end{align*}
    so $w$ would have a $d$-overlap with itself.
    Since a necessary condition for $d$ to be a period is that $d$ divides $n$, $d \leq n/2 \leq k$, and this contradicts the fact that $C$ is $(1,k)$-overlap-free.
    We will now explain that all the cyclic shifts are distinct when $k \geq n/2$.
    Suppose that in $C$, there exist words $u$ and $v$, such that the shift of $u$ for $i$ positions to the left and the shift of $v$ for $j$ positions to the left are the same, i.e.,
    \begin{align*}
        u_{i+1}\cdots u_n u_1 \cdots u_i = v_{j+1}\cdots v_n v_1 \cdots v_j.
    \end{align*}
    Without loss of generality we can assume $j \geq i$. If $j = i$, then also $u=v$, and the equality has to hold. If $j > i$, we can distinguish two cases, depending on the value of $j-i$. If $j-i \leq k$, $v_j \cdots v_{j-i}$ is a suffix of $u$. Then $u$ and $v$ have a $j-i$-overlap, but this contradicts the fact that $C$ is $(1,k)$-overlap-free.
    If $j-i > k$, then $u_1\cdots u_{n-j+i}$ is a suffix of $v$. We observe that the length of this overlap is smaller than $k$. In particular, $k \geq n/2$ implies $n-j+i < n-k \leq k$. Thus, all the shifts must be distinct. Hence, $n \lvert C  \rvert \leq \sum_{d \mid n}{\mu(d)} q^{n/d}$.
\end{proof}

\section{Constructions}\label{section:constructions}
This section gives constructions for $(t_1,t_2)$-overlap-free codes.
We first generalise the idea of Construction~\ref{construction:fimmel} to construct $(1,k)$-overlap-free codes. Then, we show that for $k \geq n/2 $, the construction yields all maximal $(1, k)$-overlap-free codes and determine sufficient conditions for its maximality.

\begin{construction}\label{(1,k)-overlap}
    Let $(L_1, R_1)$ be a partition of $\Sigma$ into two non-empty parts, and for every $i \in \left[2,k\right]$, $(L_i, R_i)$ a partition of $\bigcup_{j=1}^{i-1} \left(L_j R_{i-j} \right)$.
    Then
    \begin{align*}
        C = \bigcup_ {\substack{s \in \left[k+1,n\right] \\ j \in \left[s-k,k\right]\\\alpha \emph{ composition of } n-s\\i \in \left[0,par(\alpha)\right]}} L_{\alpha_1} \cdots L_{\alpha_i}L_jR_{s-j}R_{\alpha_{i+1}}\cdots R_{\alpha_{par(\alpha)}},
    \end{align*}
    is a $(1,k)$-overlap-free code.
\end{construction}

\begin{proof}
Any prefix $u'$ of $u \in C$ of length at most $k$ is a prefix in the set $L_{\alpha_1} \cdots L_{\alpha_{i_1}}L_{j_1}R_{s_1-j_1}$. Any suffix $v'$ of $v \in C$ of length at most $k$ is a suffix in the set $L_{j_2}R_{s_2-j_2}R_{\beta_{i_2+1}}\cdots R_{\beta_{par(\beta)}}$. If $u' = v'$, then there is a suffix in $R_{\beta_{par(\beta)}}$ that occurs either as a prefix in $R_{{s_1}-j_1}$ or as a prefix or a word in some $L_i$ for $i \in \{\alpha_1,\dots,\alpha_{i_1},j_1\}$. This cannot happen due to Proposition~\ref{proposition:z}.
\end{proof}

In the proof of Corollary~\ref{size:1k} below, we will see that Proposition~\ref{proposition:z} asserts that for every word in a code from Construction~\ref{(1,k)-overlap} there is a unique composition $\alpha$, and indices $i$ and $j$ that determine the codeword. Hence, the size of the code can be computed as follows.

\begin{corollary}\label{size:1k}
    Let $C$ be a code obtained from Construction~\ref{(1,k)-overlap}. Then
    \begin{align*}
        \lvert C \rvert = \sum_{\substack{s \in \left[k+1,n\right] \\ j \in \left[s-k,k\right]\\\alpha \emph{ composition of } n-s\\i \in \left[0,par(\alpha)\right]}} \lvert L_{\alpha_1}  \rvert \cdots \lvert L_{\alpha_i} \rvert \lvert L_j \rvert \lvert R_{s-j} \rvert  \lvert R_{\alpha_{i+1}} \rvert \cdots  \lvert R_{\alpha_{par(\alpha)}} \rvert.
    \end{align*}
\end{corollary}

\begin{proof}
    Suppose that for a word $w\in C$ there exist two compositions $\alpha$ and $\beta$, and integers $a$, $b$, $i$, $j$, such that
    \begin{align*}
        w \in L_{\alpha_1} \cdots L_{\alpha_a} L_{n-\lvert \alpha \rvert - i} R_i R_{\alpha_{a+1}}\cdots R_{\alpha_{par(\alpha)}} \cap L_{\beta_1} \cdots L_{\beta_b} L_{n-\lvert \beta \rvert - j} R_j R_{\beta_{b+1}}\cdots R_{\beta_{par(\beta)}}.
    \end{align*}
    We will show that then $\alpha = \beta$, $a = b$, and $i=j$.
    Suppose, for the sake of contradiction, a part exists on which $\alpha$ and $\beta$ differ. First, observe the case when they differ on some part with index at most $\min \{a,b\}$. Denote the smallest such index with $m$. Without loss of generality, we may assume $\alpha_m < \beta_m$. Depending on the size of $\beta_m$, there are two cases to be considered.\\
    If $\beta_m \leq \alpha_m + \cdots + \alpha_a$, there exists some part $\alpha_t$, such that some suffix of the subword of $w$ corresponding to the part $\beta_m$, i.e.,
    \begin{align*}
        u = w_{\beta_1 + \cdots + \beta_{m-1}+1} \cdots w_{\beta_1+\cdots+\beta_m} \in L_{\beta_m}
    \end{align*}
    is either a prefix or a word in $L_{\alpha_t}$, violating Proposition~\ref{proposition:z}. In the second case, when $\beta_m > \alpha_m + \cdots + \alpha_a$, we argue that there exists a suffix of $u$ that is a prefix in a word from $L_{n-\lvert \alpha \rvert -i} R_i$. For this, it is sufficient to observe that Construction~\ref{(1,k)-overlap} requires
    \begin{align*}
        \beta_1 + \cdots + \beta_b \leq n - k - 1.
    \end{align*}
    Therefore, the subword of $w$ corresponding to $L_{n-\lvert \alpha \rvert -i} R_i$ is strictly longer than $u$, since
    \begin{align*}
        n - \lvert \alpha \rvert - i + i = n - \lvert \alpha \rvert \geq k + 1 > n/2 > n - k - 1,
    \end{align*}
    and indeed there is a suffix of $u$ that is a prefix in a word from the set $L_{n-\lvert \alpha \rvert -i} R_i$, again violating Proposition~\ref{proposition:z}.
    Hence, $(\alpha_1, \cdots, \alpha_{\min\{a,b\}}) = (\beta_1, \cdots, \beta_{\min\{a,b\}})$.

    Now we are going to show that $a = b$. Without loss of generality assume $a \leq b$, and look at the suffix of the word $w$ that starts with a word from $L_{n-\lvert \alpha \rvert-i}$, i.e.,
    \begin{align*}
        v = w_{\alpha_1 + \dots + \alpha_{a}+1} \cdots w_n \in L_{n-\lvert \alpha \rvert-i}R_i R_{\alpha_{a+1}}\cdots R_{\alpha_{par(\alpha)}}.
    \end{align*}
    At the same time $v \in L_{\beta_{a+1}} \cdots L_{\beta_b} L_{n-\lvert \beta \rvert - j} R_j \cdots R_{\beta_{par(\beta)}}$. If $a \neq b$ or $n-\lvert \alpha \rvert - i \neq n - \lvert \beta \rvert - j$, it must hold that
    \begin{align*}
        \beta_{a+1}+\dots + \beta_b + n - \lvert \beta \rvert - j < n - \lvert \alpha \rvert -i.
    \end{align*}
    Otherwise, there is a suffix in $L_{n - \lvert \alpha \rvert -i}$ that is a word or a prefix in some $L_t$, which contradicts Proposition~\ref{proposition:z}. Hence, the part of the word $v$ that corresponds to $R_j$,
    \begin{align*}
        z = w_{\beta_1 + \cdots \beta_b + n-\lvert \beta  \rvert - j + 1}\cdots w_{\beta_1 + \cdots \beta_b + n-\lvert \beta \rvert},
    \end{align*}
    begins with a suffix from the set $L_{n - \lvert \alpha \rvert -i}$ which again contradicts Proposition~\ref{proposition:z}. Here, we must argue that the whole $z$ cannot be an inner word in $L_{n - \lvert \alpha \rvert -i}$ and cannot be a suffix in this set. To see this, we again use the fact from Construction~\ref{(1,k)-overlap} that the word corresponding to $L_{n-\lvert \beta \rvert -j} R_j$ is at least $k+1$ letters long, and therefore cannot be a subword in $L_{n-\lvert \alpha \rvert -i}$ with $n-\lvert \alpha \rvert -i \leq k$.
    
    Now, we may assume that $\alpha$ and $\beta$ disagree on $(\alpha_{a+1}, \cdots, \alpha_{par(\alpha)})$ and $(\beta_{b+1}, \cdots, \beta_{par(\beta)})$. In this case, we start comparing the parts from the right-most index and continue until we find the first one that differs. Following the same reasoning as before, we first observe that the last $\min \{par(\alpha)-a, par(\beta)-b\}$ parts of both compositions must agree, and then that $par(\alpha)-a = par(\beta)-b$. These observations complete the proof.
\end{proof}

\begin{proposition}
    The set of codes given by Construction~\ref{(1,k)-overlap} for $k\geq n/2$ contains all maximal $(1,k)$-overlap-free codes.
\end{proposition}
\begin{proof}
    Let $C$ be a maximal $(1,k)$-overlap-free code.
    Use $C$ to construct sets
    \begin{align*}
        L_1 &\coloneqq \{x \in \Sigma \mid \text{there exists a word in $C$ that starts in $x$} \}, \\
        R_1 &\coloneqq \Sigma \setminus L_1,
    \end{align*}
    and for $i \in \{2, \dots, k\}$
    \begin{align*}
        L_i &\coloneqq \{x \in \bigcup_{j=1}^{i-1} L_jR_{i-j} \mid \text{there exists a word in $C$ that starts in $x$} \}, \\
        R_i &\coloneqq \bigcup_{j=1}^{i-1} L_jR_{i-j} \setminus L_i.
    \end{align*}
    We use these sets to construct a $(1,k)$-overlap-free code $\hat{C}$ with Construction~\ref{(1,k)-overlap}. In the remaining part of the proof, we argue that any $w \in C$ lies in $\hat{C}$. Hence, since $C$ is maximal, $\hat{C} = C$.

    First, observe that the sets $L_i$ and $R_i$ are disjoint for every $i \leq k$ since $C$ is $(1,k)$-overlap-free. The family of pairs of sets  $(L_i,R_i)_{i < k}$, thus, satisfies all the requirements needed to apply the sequence of decompositions $(p_m(w))_{m \geq 0}$ to $w \in L_1 (L_1\cup R_1)^{*} R_1$.
    
    Since $C$ is $(1,k)$-overlap-free, every word in $C$ starts in a letter from $L_1$ and ends in a letter from $R_1$. Hence, $C \subseteq L_1 (L_1\cup R_1)^{*} R_1$, and the sequence of decompositions $(p_m(w))_{m \geq 0}$ is defined for every $w \in C$.
    Every element $p_m(w)$ must start in l and end in r; otherwise, for some $1 \leq t \leq k$, there is a $t$-overlap between two words from $C$.

    Fix $w \in C$ and denote the last element of the sequence of decompositions $(p_m(w))_{m \geq 0}$, where the decompositions $p_m(w)$ of $w$ stabilise under the reduction process, by $p_{\hat{m}}(w)$.
    We claim that there is at most one occurrence of lr in $p_{\hat{m}}(w)$. If there were at least two occurrences of lr, then each would correspond to a subword longer than $k \geq n/2$ and the length of $w$ would exceed $n$. Therefore $p_{\hat{m}} \in l^+r^+$, and $w \in \hat{C}$.
\end{proof}

\begin{proposition}\label{proposition:maximal}
    A code $C$ obtained from Construction~\ref{(1,k)-overlap} is maximal, if every $x \in L_t$ is a prefix in  $C$ and every $x \in R_t$ a suffix in $C$ for all $1 \leq t \leq k$. If $\lvert L_{n/2} \cup R_{n/2}\rvert > 1$, the reverse also holds.
\end{proposition}

\begin{proof}
    $(\Leftarrow):$ \\
    Suppose that $C$ is not maximal. Then there exists some $w \in \Sigma^n \setminus C$ such that $C \cup \{w\}$ is a $(1,k)$-overlap-free code. Form a finite sequence of decompositions $p_1(w), p_2(w), \dots \in \{\text{l}, \text{r}\}^+$ as before, and consider the last element in this sequence. Since $k \geq n/2$, the last element of the  decomposition is of one of the forms l\textsuperscript{+}, r\textsuperscript{+}, l\textsuperscript{+}r\textsuperscript{+}, r\textsuperscript{+}l\textsuperscript{+}, l\textsuperscript{+}r\textsuperscript{+}l\textsuperscript{+}, r\textsuperscript{+}l\textsuperscript{+}r\textsuperscript{+}. Since $w \not\in C$, it cannot be l\textsuperscript{+}r\textsuperscript{+}. It also cannot be of the forms l\textsuperscript{+}, r\textsuperscript{+}l\textsuperscript{+} nor l\textsuperscript{+}r\textsuperscript{+}l\textsuperscript{+} since then $w$ ends in a word from some $L_t$ that occurs as a prefix in $C$. If the last element of the decomposition belongs to any of the forms r\textsuperscript{+}, r\textsuperscript{+}l\textsuperscript{+} or r\textsuperscript{+}l\textsuperscript{+}r\textsuperscript{+}, $w$ starts in a word from some $R_t$ that is a $t \leq k$ letters long suffix in $C$. Hence, the assumption cannot hold, and $C$ is maximal.
    
    $(\Rightarrow):$ \\
    By symmetry, we only need to settle the case for prefixes. Suppose an index $t$ and $x \in L_t$ exist such that $x$ is not a prefix in $C$. We show that then $C$ is not maximal.
    In particular, if $t = k$, then $L_1^{n-t}x \cup C$ is $(1,k)$-overlap-free. If $t \leq n-k-1$, then $Y = L_1^{n-k-t}R_kx \cup L_1^{n-k-1-t}L_lR_1x$ is non-empty and $C \cup Y$ $(1, k)$-overlap-free.
    If $n-k-1 < t < k$, then $n-t < k +1$. The sets $L_{n-t}$ and $R_{n-t}$ are defined, and their union is non-empty. Since $xR_{n-t}\subseteq C$ and $x$ is not a prefix in $C$, the set $R_{n-t}$ is empty. But then $L_{n-t}x$ is non-empty. Unless $t=n-t = n/2$, the set $C \cup L_{n-t}x$ is $(1,k)$-overlap-free and larger than $C$. For $t=n-t = n/2$, we need the condition $\lvert R_{n/2} \rvert + \lvert L_{n/2} \rvert > 1$. Since we already explained that $R_{n/2}$ is empty, the requirement implies $\lvert L_{n/2} \rvert > 1$. Then the set $L_{n-t}\setminus\{x\}$ is non-empty, so $C \cup \left(L_{n-t}\setminus\{x\}\right)x$ is larger than $C$, but still $(1,k)$-overlap-free.
\end{proof}

The case $\lvert L_{n/2} \cup R_{n/2} \rvert = 1$ only happens over a binary alphabet. Additionally, $\lvert L_i \cup R_i \rvert = 1$ holds for $1 < i < n/2$ in this case. In particular, either $\lvert L_2 \rvert = \cdots = \lvert L_{n/2-2} \rvert = 0$ or $\lvert R_2 \rvert = \cdots = \lvert R_{n/2-2} \rvert = 0$~\cite{stanovnik:2024}. Hence, we can further characterise maximal $(1,k)$-overlap-free codes.

\begin{proposition}
    Suppose $C$ from Construction~\ref{(1,k)-overlap} is maximal and $\lvert L_{n/2} \cup R_{n/2} \rvert = \{u\}$. The following statements hold.\\
    (i) For all $t$ distinct from $n/2$, every $x \in L_t$ is a prefix in $C$ and every $x \in R_t$ a suffix in $C$. \\
    (ii) If $u \in L_{n/2}$ and $u$ is not a prefix in $C$, then for every non-empty $L_j, 2\leq j \leq k - n/2$, the word in $L_jL_{n/2}$ is a prefix in $C$. Moreover, if $k > n/2$, the word in $L_1L_{n/2}$ is a prefix in $C$, or $L_{n/2-1}$ is empty and for every non-empty $L_j$, $1 \leq j \leq k - n/2 -1$, the word in $L_j L_1L_{n/2}$ is a prefix in $C$. \\
    (iii) If $u \in R_{n/2}$ and $u$ is not a suffix in $C$, then for every non-empty $R_j, 2\leq j \leq k - n/2$, the word in $R_{n/2}R_j$ is a suffix in $C$. Moreover, if $k > n/2$, the word in $R_{n/2}R_1$ is a suffix in $C$, or $R_{n/2-1}$ is empty and for every non-empty $R_j$, $1 \leq j \leq k - n/2 -1$, the word in $R_{n/2}R_1R_j$ is a suffix in $C$.
\end{proposition}
\begin{proof}
    Statement (i) follows from the proof of Proposition~\ref{proposition:maximal}. Due to the symmetry of (ii) and (iii), we only prove (ii).
    Suppose that for some $j \in \left[k-n/2\right]$, the word $v$ in $L_jL_{n/2}$ is not a prefix in $C$. 
    Since $n/2-j < k$, the sets $R_{n/2-j}$ and $L_{n/2-j}$ are defined and their union is non-empty. If $L_{n/2-j}$ is non-empty, $C \cup L_{n/2}L_jL_{n/2}$ is a $(1,k)$-overlap-free code larger than $C$. This contradicts the maximality of $C$. On the other hand, if $R_{n/2-j}$ is non-empty, $L_jL_{n/2}R_{n/2-j}$ is a non-empty subset of $C$ that contains a word that starts in $v$ which cannot happen if $v$ is not a prefix in $C$. Since the sets $L_{n/2-j}$ and $R_{n/2-j}$ cannot be simultaneously empty, the initial assumption cannot hold. 

    Now suppose that the word $v'$ in $L_1L_{n/2}$ is not a prefix in $C$, $L_{n/2-1}$ is empty, and there exists some $j \leq k - n/2 - 1$ such that the word in $L_jL_1L_{n/2}$ is not a prefix in $C$. Since $n/2-1-j < k$, the sets $L_{n/2-1-j}$ and $R_{n/2-1-j}$ are defined and their union is non-empty. If $L_{n/2-1-j}$ is non-empty, $C \cup L_{n/2-1-j}L_jL_1L_{n/2}$ is a $(1,k)$-overlap-free code larger than $C$. This cannot hold as $C$ was chosen to be maximal. Hence, $L_{n/2-1-j}$ must be empty. If $R_{n/2-1-j}$ is non-empty, $L_jL_1L_{n/2}R_{n/2-1-j}$ is a non-empty subset of $C$ and there is a word in $C$ that starts in $v'$. This implies that the union of $R_{n/2-1-j}$ and $L_{n/2-1-j}$ is empty, which cannot hold due to Construction~\ref{(1,k)-overlap}. Therefore, our assumption was wrong, and statement (ii) holds.
\end{proof}

Before we continue constructing $(t_1,t_2)$-overlap-free codes, we want to recall a known construction of weakly mutually uncorrelated codes from non-overlapping codes. 
Yazdi et al.~\cite{Yazdi:2018} explained that concatenating a non-overlapping code $X \subseteq \Sigma^n$ with $\Sigma^k$  yields a $(k+1)$-weakly mutually uncorrelated code of length $n+k$.
Construction~\ref{construction:kWMU} shows that their construction is often expandable, even if $C$ is a maximal non-overlapping code.

\begin{construction}\label{construction:kWMU}
    Let $X = \bigcup_{i=1}^{n-1} L_i R_{n-i}$ be a code obtained from Construction~\ref{construction:fimmel}.
    Then
    \begin{align*}
        C = \bigcup_{n \leq i+j \leq n+k} L_i R_j\Sigma^{n+k-i-j}
    \end{align*}
    is a $(k+1)$-weakly mutually uncorrelated code of length $n+k$.
\end{construction}

\begin{proof}
    Every $l$-letter suffix of $C$ with $l > k$ has a prefix $w$ that is a suffix in some $L_iR_j$. Its length equals $l-(n+k-i-j) \geq i+j+1-n \geq 1$ . Suppose $w$ is a prefix in some $L_{p}R_{r}\Sigma^{n+k-p-r}$. 
    Construct (any) partitions $(L_s, R_s)$  for $s \in \left[n,n+k-1\right]$ with the method from Construction~\ref{construction:fimmel}, and denote $Y = \bigcup_{s \in \left[n+k-1\right]} L_sR_{n+k-s}$.
     If $\lvert w \rvert < p + r$, word $w$ is a suffix in $L_{i+j} \cup R_{i+j}$ and a prefix in $L_{p+r} \cup R_{p+r}$. Since $Y$ is non-overlapping, this contradicts Proposition~\ref{proposition:z}.
     If $\lvert w \rvert \geq p + r$,  then $\lvert w \rvert \geq n$ and $w$ has a prefix $\hat{w}$ that is a suffix in $L_i$ because $j < n$. Depending on its length, either $\hat{w}$ is a prefix in $L_{p}$ or has a suffix that occurs as a prefix in $R_{r}$. Both cases contradict Proposition~\ref{proposition:z}. Hence, $C$ is $(k+1)$-weakly mutually uncorrelated.
\end{proof}

We claim that for $(i,j) \neq (p, r)$, the sets of words $L_iR_j\Sigma^{n+k-i-j}$ and $L_{p}R_{r}\Sigma^{n+k-p-r}$ are disjoint. To see this, note that if $i+j = p + r$, the reverse would contradict Proposition~\ref{proposition:unique_index}. For $i+j \neq p + r$, we can assume $i \leq p$ without loss of generality. If $i = p$, a non-empty intersection of $L_iR_j\Sigma^{n+k-i-j}$ and $L_{p}R_{r}\Sigma^{n+k-p-r}$ would imply that a word from the set $R_{\min\{j,r\}}$ is a prefix in a word from $R_{\max \{j, r\}}$ contradicting Proposition~\ref{proposition:z} and establishing our claim. If $i < p$, the non-empty intersection would mean that a suffix from $L_{p}$ occurs as a prefix in $R_j$ (since $i+j \geq n$ and $p < n$), again contradicting Proposition~\ref{proposition:z} and establishing our claim. Hence, the size of the code $C$ generated by Construction~\ref{construction:kWMU} equals
\begin{align*}
    \lvert C \rvert = \sum_{n \leq i+j \leq n+k}\lvert L_i \rvert \lvert R_{j} \rvert q^{n+k-i-j}.
\end{align*}

This idea can be generalised to construct a $(t_1,t_2)$-overlap-free code. We first provide a simple construction based on the method from Yazdi et al.~\cite{Yazdi:2018} that will be used in the next section to determine a lower bound on $S_{t_1}^{t_2}(q,n)$, and then expand it using the analogy with weakly mutually uncorrelated codes.
\begin{construction}\label{construction:t1t2}
Let $X$ be a $(1,t_2)$-overlap-free code of length $n-t_1+1$ if $t_2 < n - t_1 + 1$ and a non-overlapping code of length $n-t_1+1$, otherwise. Then $C \coloneqq X\Sigma^{t_1 - 1}$ is a $(t_1,t_2)$-overlap-free code of length $n$. The number of words in it equals $\lvert C \rvert = \lvert X \rvert q^{t_1 - 1}$.
\end{construction}

\begin{remark}
    If $t_2 > n-t_1$, Construction~\ref{construction:t1t2} provides a $t_1$-weakly mutually uncorrelated code that is obtainable from the Construction from Yazdi et al.~\cite{Yazdi:2018}. Thus, its expansion is the same as Construction~\ref{construction:kWMU}, and thus we do not explicitly write it down in Construction~\ref{construction:expandedt1t2}.
\end{remark}

\begin{construction}\label{construction:expandedt1t2}
Let $t_1 + t_2 \leq n$. Let $X$ be a $(1,t_2)$-overlap-free code of length $n-t_1+1$ from Construction~\ref{(1,k)-overlap}.
\begin{align*}
    C = \bigcup_{\substack{k\in\left[0,t_1-1\right] \\ s \in \left[t_1+t_2-k,n-k\right] \\ j \in \left[s-t_2,t_2\right] \\\alpha \text{ composition of } n-k-s\\ i \in \left[0, par(\alpha)\right]}}
    L_{\alpha_1} \cdots L_{\alpha_i} L_jR_{s-j} R_{\alpha_{i+1}}\cdots R_{\alpha_{par(\alpha)}}\Sigma^{k},
\end{align*}
where for $i \in \left[t_2\right]$ the sets $L_i$ and $R_i$ are defined in Construction~\ref{(1,k)-overlap}, is a $(t_1,t_2)$-overlap-free code.
\end{construction}
\begin{proof}
    Let $w$ be a suffix of a word in $L_{\alpha_1} \cdots L_{\alpha_i} L_jR_{s-j} R_{\alpha_{i+1}}\cdots R_{\alpha_{par(\alpha)}}\Sigma^{k}$ with length $l$, $t_1 \leq l \leq t_2$. Then, $w$ has a prefix $w'$ that is a suffix in $L_{\alpha_1} \cdots L_{\alpha_i} L_jR_{s-j} R_{\alpha_{i+1}}\cdots R_{\alpha_{par(\alpha)}}$. Because $l-k \leq t_2 -k$, $w'$ is a suffix in $L_jR_{s-j} R_{\alpha_{i+1}}\cdots R_{\alpha_{par(\alpha)}}$. Suppose, for sake of contradiction, that $w'$ is a prefix in $L_{\beta_1}\cdots L_{\beta_{i'}}L_{j'}R_{s'-j'}R_{\beta_{i'+1}}\cdots R_{\beta_{par{\beta}}} \Sigma^{k'}$. Since $l-k \leq t_2 - k$, $w'$ is a prefix in $L_{\beta_1}\cdots L_{\beta_{i'}}L_{j'}R_{s'-j'}$. Now, there exist integers $m$ and $m'$ and, depending on the length of $w'$, a prefix or a suffix of $w'$ that is a prefix in $L_m \cup R_m$ and a suffix in $L_{m'}\cup R_{m'}$, contradicting Proposition~\ref{proposition:z}.
\end{proof}

\noindent Finally, we provide a construction of simultaneously $(1,k)$- and $(n-k,n-1)$-overlap-free codes for $k < n/2$.
\begin{construction}\label{construction:simmultaneous}
Let $X = \bigcup_{i\leq k} L_iR_{k+1-i}$ be a non-overlapping code of length $k+1$.
\begin{align*}
    C = \bigcup_{\alpha} X \Sigma^{n-2k-1} R_{\alpha_1}\cdots R_{\alpha_{par(\alpha)}} ,
\end{align*}
where $\alpha$ iterates over all compositions of $k$, is a simultaneously $(1,k)$-overlap-free and $(n-k,n-1)$-overlap-free code.
\end{construction}

\begin{proof}
    Any suffix $u$ of $C$ of length at least $n-k$ has a prefix $u'$ that occurs as a suffix in $X$. Since $X$ is non-overlapping $u'$ cannot occur as a prefix in $X$, so $C$ is $(n-k,n-1)$-overlap-free. Moreover, any suffix $u$ of $C$ shorter than $k+1$ starts in a prefix $u'$, a word in some $R_i$, or a suffix of some $R_i$. None of these can occur as a prefix in $X$ due to Proposition~\ref{proposition:z}.
\end{proof}

\section{Lower bounds}\label{section:lower}
In this section, we determine lower bounds on $S_{t_1}^{t_2}$ and compare them with the upper bounds given in Section~\ref{section:upper}.
\begin{proposition}\label{proposition:lower_bound}
    \begin{align*}
        S_1^k(q,n) \geq \left(q-\biggl\lfloor\frac{k}{k+1}\cdot q\biggr\rfloor\right) \biggl\lfloor\frac{k}{k+1}\cdot q\biggr\rfloor^k q^{n-k-1}.
    \end{align*}
\end{proposition}
\begin{proof}
    We first consider the case when $k \geq n/2$.
    Let $k < n \leq 2k$, and $x$ be some integer between 0 and $q$. Partition $\Sigma$ into two sets of sizes $x$ and $q-x$ respectively. Denote the set with $x$ elements $R_1$ and the set with $q- x$ elements $L_1$. Now proceed with Construction~\ref{(1,k)-overlap} by setting $R_i$ empty for every $i > 1$. Observe that this implies $L_i = L_1R_1^{i-1}$ for $2 \leq i \leq k$. The resulting code $(1,k)$-overlap-free code $C$ equals
    \begin{align*}
        C = \bigcup_{l=0}^{n-k-1}\bigcup_{i=0}^l \bigcup_{\alpha} L_{\alpha_1} \cdots L_{\alpha_i}L_1 R_1^k R_1^{\alpha_{i+1} + \cdots +\alpha_l},
    \end{align*}
    where $\alpha$ iterates over all compositions of $n-k-1$ into $l$ parts. 
    Introduce $j = \alpha_{i+1} + \cdots +\alpha_l$ and rewrite the unions so that $j$ iterates over all integers from 0 to $n-k-1$ and $\alpha$ iterates over all compositions of $n-k-1-j$. In particular,
    \begin{align*}
        C = \bigcup_{j=0}^{n-k-1}\bigcup_{l=0}^{n-k-1-j} \bigcup_{\alpha} \left(L_1R_1^{\alpha_1-1}\right) \cdots \left(L_1R_1^{\alpha_l-1}\right)\left(L_1 R_1^k \right) R_1^{j}.
    \end{align*}
    The size of this code can be determined from Corollary~\ref{size:1k}.
    \begin{align*}
        \lvert C \rvert &= \sum_{j=0}^{n-k-1} \sum_{l=0}^{n-k-1-j}\sum_{\alpha} \lvert L_1\rvert^{l+1} \lvert R_1\rvert^{n-k-1-j-l+k+j} \\
        &= \sum_{j=0}^{n-k-1} \sum_{l=0}^{n-k-1-j}\sum_{\alpha} (q-x)^{l+1}x^{n-1-l}.
    \end{align*}
    For $j = n - k - 1$ there is one composition of $n-k-1-j$ into zero parts and no compositions into $l > 0$ parts.
    For $j < n - k - 1$ there are no compositions of $n-k-1-j$ into zero parts, and the number of compositions into $l > 0$ parts equals $\binom{n-k-j-2}{l-1}$. Hence, 
    \begin{align*}
        \lvert C \rvert &= (q-x) x^{n-1} + \sum_{j=0}^{n-k-2}\sum_{l=1}^{n-k-1-j}\binom{n-k-j-2}{l-1} (q-x)^{l+1} x^{n-1-l},
    \end{align*}
   and after performing additional calculations one obtains the simplified formula for $\lvert C \rvert$.
    \begin{align*}
        \lvert C \rvert 
        &= (q-x)x^{n-1} + (q-x)^2 x^k \sum_{j=0}^{n-k-2} x^j (q-x+x)^{n-k-2-j} \\
        &= (q-x)x^{n-1} + (q-x)^2 x^k q^{n-k-2} \sum_{j=0}^{n-k-2}  \left(\frac{x}{q}\right)^j \\
        &= (q-x)x^{n-1} + (q-x)^2 x^k q^{n-k-2} \cdot \frac{1 - \left(\frac{x}{q}\right)^{n-k-1}}{1-\frac{x}{q}} \\ 
        &= (q-x)x^{n-1} + (q-x)x^k(q^{n-k-1} - x^{n-k-1}) \\
        &= (q-x)x^k q^{n-k-1}.
    \end{align*}
    We are now left to find the integer value $x \in (0, q)$ that produces the largest code $C$. We first analyse $f(x) = (q-x)x^k q^{n-k-1}$ on the closed interval $\left[0,q\right]$. Since $f$ is differentiable, $f(0) = f(q) = 0$ and $f(x) > 0$ for $x \in (0,q)$, it follows from Fermat's theorem that the global maxima of $f$ lies at a stationary point in the interior of the interval. Its derivative
    \begin{align*}
        f'(x) = q^{n-k-1}x^{k-1}(qk - (k+1)x),
    \end{align*}
    has exactly one stationary point on the open interval $(0,q)$, the point $x = \frac{k}{k+1} \cdot q$ in which $f$ achieves its maximum.
    Therefore, $f$ restricted to integers attain its maximum either at $\bigl\lfloor \frac{k}{k+1} \cdot q \bigr\rfloor$ or $\bigl\lceil \frac{k}{k+1} \cdot q \bigr\rceil$. In the bound, we round the value down since $1 \leq \bigl\lfloor \frac{k}{k+1} \cdot q \bigr\rfloor \leq q-1$, while for $q < k+1$, $\bigl\lceil \frac{k}{k+1} \cdot q \bigr\rceil > q -1$ giving an invalid partition for our construction.
    
    \noindent We are left to cover the case $k < n/2$. We know from Eq.~\eqref{eq:4} that
    \begin{align*}
        S_1^k(q,n) &= q^{n-2k} S_1^k(q,2k) \\
        &\geq  q^{n-2k+2k-k-1} \left(q - \biggl\lfloor \frac{k}{k+1} \cdot q\biggr\rfloor \right) \biggl\lfloor\frac{k}{k+1}\cdot q\biggr\rfloor^k \\
        &= q^{n-k-1} \left(q - \biggl\lfloor\frac{k}{k+1}\cdot q\biggr\rfloor\right) \biggl\lfloor\frac{k}{k+1}\cdot q\biggr\rfloor^k.
    \end{align*}
    This shows that the proposition holds for all values of $1 \leq k < n$.
\end{proof}

We want to remark that for $k=2$ the bound in Proposition~\ref{proposition:lower_bound} is the same as the exact value given in Eq.~\eqref{eq:5}. Moreover, when $k+1$ divides $q$, the lower bound equals $\left(\frac{k}{k+1}\right)^k \frac{q^n}{k+1}$ and we will show that the upper bound on $S_1^k(q,n)$ from Section~\ref{section:upper} is tight to within a constant factor.

\begin{proposition}
    If $k+1$ divides $q$
    \begin{align*}
        S_1^k(q,n) \geq \frac{q^n}{e n},
    \end{align*}
    where $e$ is the base of the natural logarithm.
\end{proposition}

\begin{proof}
By Proposition~\ref{proposition:lower_bound} and since $k+1 \leq n$, we obtain
    \begin{align*}
        S_1^k(q,n) &\geq \left(\frac{k}{k+1}\right)^k \frac{q^n}{k+1} \\
        &\geq \left(\frac{k+1}{k}\right)^{-k} \frac{q^n}{n}.
    \end{align*}
    From the fact that $\left(\frac{k+1}{k}\right)^{k} = \left(1 +\frac{1}{k}\right)^{k} < e$ for every positive $k$, we know that $\left(\frac{k+1}{k}\right)^{-k} > e^{-1}$. The bound follows.
\end{proof}

Even when $k+1$ does not divide $q$, but $q$ is large, the lower bound we determined is close to a constant factor of the upper bound, in the sense that
\begin{align*}
    S_1^k(q,n) \geq \frac{q^n}{en} - O(q^{n-1}).
\end{align*}
However, for a small value of $q$, the bound we determine can be much lower than the known upper bound. If $q < k + 1 \leq n$, $\frac{k}{k+1} q > q - 1$, and hence the formula on the right-hand side in Proposition~\ref{proposition:lower_bound} achieves a value zero if $\frac{k}{k+1} q$ is rounded to $q$, and a non-zero value if we round $\frac{k}{k+1} q$ to $q-1$. Thus,
\begin{align*}
    S_1^k(q,n) &\geq (q-1)^k q^{n-k-1} \\
    & \geq \left(\frac{q-1}{q}\right)^k \frac{q^n}{q} \\
    & \geq \left(\frac{q-1}{q}\right)^k \frac{q^n}{n}.
\end{align*}

 In the case $k=n-1$, we know that the lower bound we determined in Proposition~\ref{proposition:lower_bound} is also the upper bound on the size of a $(1,n-1)$-overlap-free code given in Eq.~\eqref{eq:1}. It would be interesting to know whether this result can be generalised, i.e., whether, for all values of $k$, it holds that $S_1^k(q,n) \leq \left(\frac{k}{k+1}\right)^k \frac{q^n}{k+1}$. We remark that this conjecture cannot be proven by upper-bounding $S_1^k(q,n)$ by $S_1^k(q,k+1) q^{n-k-1}$ and applying Eq.~\eqref{eq:1}. More specifically, for $k \in \{n-2,n-3\}$, we verify in Proposition~\ref{proposition:lower_combinatorial} that $S_1^k(q,n) \geq S_1^k(q,k+1) q^{n-k-1}$, and for some values given in Tables~\ref{tab:(1,n-2)-overlap-free} and \ref{tab:(1,n-3)-overlap-free} a strict inequality holds. They are written in bold. Still, the computed values are all smaller than $\left(\frac{k}{k+1}\right)^k \frac{q^n}{k+1}$.

\begin{table}[ht]
    { %
    \centering %
    \begin{tabular}{rrrr} \hline
         $n$ & $q=2$ & $q=3$ & $q=4$ \\\hline
         5 & \textbf{3} & 24 & 108 \\ 
         6 & \textbf{5} & \textbf{52} & 324 \\
         7 & \textbf{8} & \textbf{140} & \textbf{1,080} \\
         8 & \textbf{13} & \textbf{338} & \textbf{3,567} \\
         9 & \textbf{21} & \textbf{876} & \textbf{12,558} \\
         10 & \textbf{37} & \textbf{2,248} & \textbf{44,726} \\
         11 & \textbf{68} & \textbf{6,688} & \textbf{159,294} \\
         12 & \textbf{125} & \textbf{18,272} & \textbf{567,334} \\
         13 & \textbf{230} & \textbf{49,920} & \textbf{2,066,100} \\
         14 & \textbf{423} & \textbf{136,384} & \textbf{7,494,499} \\
         15 & \textbf{778} & \textbf{372,608} & \textbf{27,858,764} \\
         16 & \textbf{1,431} & \textbf{1,021,312} & \textbf{102,773,636} \\
         17 & \textbf{2,632} & \textbf{2,864,088} & \textbf{397,629,405} \\\hline
    \end{tabular}
    \caption{The sizes of $(1,n-2)$-overlap-free codes constructed from a maximum non-overlapping code of length $n-1$ determined in \cite{stanovnik:2024} using Construction~\ref{(1,k)-overlap}.}
    \label{tab:(1,n-2)-overlap-free}}
    \end{table}
    \begin{table}[ht]
    {
    \centering
    \begin{tabular}{rrrr} \hline
         $n$ & $q=2$ & $q=3$ & $q=4$ \\\hline
         6 & \textbf{6} & 72 & 432 \\
         7 & \textbf{10} & \textbf{156} & 1,296 \\
         8 & \textbf{16} & \textbf{420} & \textbf{4,320} \\
         9 & \textbf{26} & \textbf{1,014} & \textbf{14,268} \\
         10 & \textbf{45} & \textbf{2,725} & \textbf{50,232} \\
         11 & \textbf{79} & \textbf{6,744} & \textbf{178,904} \\
         12 & \textbf{149} & \textbf{20,064} & \textbf{637,176} \\
         13 & \textbf{274} & \textbf{54,816} & \textbf{2,269,336} \\
         14 & \textbf{504} & \textbf{149,760 }& \textbf{8,387,491} \\
         15 & \textbf{927} & \textbf{409,152} & \textbf{30,424,492} \\
         16 & \textbf{1,705} & \textbf{1,117,824} & \textbf{114,655,998} \\
         17 & \textbf{3,136} & \textbf{3,099,715} & \textbf{422,976,906} \\
         18 & \textbf{5,768} & \textbf{8,691,344} & \textbf{1,590,517,620} \\\hline
    \end{tabular}
    \caption{The sizes of $(1,n-3)$-overlap-free codes constructed from a maximum non-overlapping code of length $n-2$ determined in \cite{stanovnik:2024} using Construction~\ref{(1,k)-overlap}.}
    \label{tab:(1,n-3)-overlap-free}
    }
\end{table}

\begin{proposition}\label{proposition:lower_combinatorial}
    Let $q \geq 2$. For all integers $n \geq 4$,
    \begin{align*}
        S_1^{n-2}(q,n) &\geq q\cdot S_1^{n-2}(q,n-1).
    \end{align*}
    Moreover, for all integers $n \geq 6$
    \begin{align*}
        S_1^{n-3}(q,n) &\geq q^2\cdot S_1^{n-3}(q,n-2).
    \end{align*}
\end{proposition}

\begin{proof}
    Let $C$ be a maximum non-overlapping code of length $k+1$. Since every maximum non-overlapping code is also maximal, $C$ can be obtained from Construction~\ref{construction:fimmel} using a collection of pairs of disjoint sets $(L_i,R_i)_{i\in \left[k\right]}$.
    \begin{align*}
        C = \bigcup_{i \in \left[k\right]} L_i R_{k+1-i}.
    \end{align*}

    If $k \geq 2$, the same collection of pairs of disjoint sets can be used to generate a $(1,k)$-overlap-free code $D_1$ of length $k+2$ with Construction~\ref{(1,k)-overlap},
    \begin{align*}
        D_1 = \bigcup_{i \in \left[k\right]} L_1 L_i R_{k+1-i} \cup \bigcup_{i \in \left[k\right]} L_i R_{k+1-i} R_1 \cup \bigcup_{i \in \left[2,k \right]} L_iR_{k+2-i}.
    \end{align*}
    The size of $D_1$ equals
    \begin{align*}
        \lvert D_1 \rvert &= \lvert L_1 \rvert \sum_{i\in \left[k\right]}  \lvert L_i \rvert \lvert R_{k+1-i}\rvert + \lvert R_1 \rvert \sum_{i\in \left[k\right]} \lvert L_i \rvert \lvert R_{k+1-i} \rvert + \sum_{i \in \left[2,k \right]} \lvert L_i \rvert \lvert R_{k+2-i} \rvert \\
        &= \left(\lvert L_1 \rvert + \lvert R_1 \rvert \right) \lvert C \rvert + \sum_{i \in \left[2,k \right]} \lvert L_i \rvert \lvert R_{k+2-i} \rvert \\
        &= q \, \lvert C \rvert + \sum_{i \in \left[2,k \right]} \lvert L_i \rvert \lvert R_{k+2-i} \rvert \\
        &\geq q \, \lvert C \rvert = q S_1^k(k+1,q).
    \end{align*}
    The size of the maximum $(1,k+2)$-overlap-free code is at least the size of $D_1$, proving the first claim of our proposition.

    We proceed with the second claim. If $k \geq 3$, the collection of pairs of disjoint sets $(L_i,R_i)_{i\in \left[k\right]}$ can be used to construct a $(1,k)$-overlap-free code $D_2$ of length $k+3$ with Construction~\ref{(1,k)-overlap},
    \begin{align*}
        D_2 = &\bigcup_{i \in \left[k\right]} (L_1^2  \cup L_2) L_i R_{k+1-i} \cup \bigcup_{i \in \left[k\right]} L_i R_{k+1-i} (R_1^2 \cup R_2) \cup \bigcup_{i \in \left[k\right]} L_1 L_i R_{k+1-i} R_1 \\
        & \cup L_1 \bigcup_{i \in \left[2,k \right]} L_iR_{k+2-i} \cup \bigcup_{i \in \left[2,k \right]} L_iR_{k+2-i} R_1 
        \cup \bigcup_{i \in \left[3,k \right]} L_iR_{k+3-i}.
    \end{align*}
    We compute its size
    \begin{align*}
        \lvert D_2 \rvert = &(\lvert L_1\rvert^2 + \lvert L_2 \rvert + \lvert R_1\rvert^2 + \lvert R_2 \rvert + \lvert L_1 \rvert \lvert R_1 \rvert) \lvert C \rvert + \\
        &(\lvert L_1 \rvert + \lvert R_1 \rvert) \sum_{i\in \left[2,k\right]} \lvert L_i \rvert \lvert R_{k+2-i} \rvert + \sum_{i \in \left[3,k\right]} \lvert L_i \rvert \lvert R_{k+3-i} \rvert.
    \end{align*}
    Knowing that $\lvert L_2 \rvert + \lvert R_2 \rvert =  \lvert L_1 \rvert \lvert R_1 \rvert$ and $\lvert L_1 \rvert + \lvert R_1 \rvert = q$, we can rewrite this size as follows
    \begin{align*}
    \lvert D_2 \rvert &= (\lvert L_1 \rvert + 2 \lvert L_1 \rvert \lvert R_1 \rvert + \lvert R_1 \rvert) \lvert C \rvert + q \sum_{i\in \left[2,k\right]} \lvert L_i \rvert \lvert R_{k+2-i} \rvert + \sum_{i \in \left[3,k\right]} \lvert L_i \rvert \lvert R_{k+3-i} \rvert \\
    &= q^2 \lvert C \rvert + q \sum_{i\in \left[2,k\right]} \lvert L_i \rvert \lvert R_{k+2-i} \rvert + \sum_{i \in \left[3,k\right]} \lvert L_i \rvert \lvert R_{k+3-i} \rvert.
    \end{align*}
    Since $S_1^k(q,k+3) \geq \lvert D_2 \rvert \geq q^2 \lvert C \rvert = q^2 S_1^k(q, k+1)$, the second claim is established.
\end{proof}

We determine a lower bound on the size of $(t_1,t_2)$-overlap-free codes using Proposition~\ref{proposition:lower_bound} and Construction~\ref{construction:t1t2}. 
\begin{corollary}
        \begin{align*}
            S_{t_1}^{t_2}(q,n) \geq q^{n-t_2-1} \left(q - \biggl\lfloor\frac{t_2}{t_2+1} \cdot q\biggr\rfloor\right) \biggl\lfloor\frac{t_2}{t_2+1} \cdot q\biggr\rfloor^{t_2},
        \end{align*}
    where 
    \begin{align*}
        k = \begin{cases}
            t_2 & \text{if } t_1 + t_2 \leq n, \\
            n-t_1 & \text{otherwise}.
        \end{cases}
    \end{align*}
\end{corollary}
\begin{proof}
If $t_2 \leq (n-t_1+1)-1$, i.e., $t_1+t_2 \leq n$, Construction~\ref{construction:t1t2} provides the code $X\Sigma^{t_1-1}$ with the largest size when $X$ is a maximum $(1,t_2)$-overlap-free code.
Otherwise, it provides the code $X\Sigma^{t_1-1}$ with the largest size when $X$ is a maximum non-overlapping code. Hence, for 
    \begin{align*}
        k = \begin{cases}
            t_2 & \text{if } t_1 + t_2 \leq n, \\
            n-t_1 & \text{otherwise},
        \end{cases}
    \end{align*}
\begin{align*}
    S_{t_1}^{t_2}(q,n) &\geq q^{t_1-1} S_{1}^{k}(q,n-t_1+1).
\end{align*}
Now, use Proposition~\ref{proposition:lower_bound}
\begin{align*}
    S_{t_1}^{t_2}(q,n) &\geq q^{t_1-1 + n-t_1+1-k-1} \left(q - \biggl\lfloor\frac{k}{k+1} \cdot q\biggr\rfloor\right) \biggl\lfloor\frac{k}{k+1} \cdot q\biggr\rfloor^{k} \\
    &= q^{n-k-1} \left(q - \biggl\lfloor\frac{k}{k+1} \cdot q\biggr\rfloor\right) \biggl\lfloor\frac{k}{k+1} \cdot q\biggr\rfloor^{k},
\end{align*}
to conclude the proof.
\end{proof}
When $t_1 + t_2 \leq n$ and $t_2 + 1$ divides $q$, 
\begin{align*}
        S_{t_1}^{t_2}(q,n) &\geq \left(\frac{t_2}{t_2+1}\right)^{t_2} \frac{q^n}{t_2+1}\\
        &\geq e^{-1} \cdot \frac{q^n}{t_2+1},
\end{align*}
and the upper bound on the size of $(t_1,t_2)$-overlap-free codes from Section~\ref{section:upper} is tight to up to a factor $\bigl\lfloor \frac{2n-t_1}{n-t_2}\bigr\rfloor\cdot((t_2+1)e)^{-1}$. For large values of $q$, the lower bound is also close to this factor, but the gap may be much greater when $q$ is small. In particular, if $q < t_2 + 1$,
\begin{align*}
    S_{t_1}^{t_2}(q,n) &\geq q^{n-t_2-1}(q-1)^{t_2} \\
    &\geq \left(\frac{q-1}{q}\right)^{t_2} \frac{q^n}{q} \\
    &\geq \left(\frac{q-1}{q}\right)^{t_2} \frac{q^n}{t_2 + 1}.
\end{align*}

Lastly, we derive a lower bound on the size of a simultaneously $(1,k)$-overlap-free and $(n-k,n-1)$-overlap-free code from Construction~\ref{construction:simmultaneous}. The proof also shows that the maximum number of words in a simultaneously $(1,1)$-overlap-free and $(n-1,n-1)$-overlap-free code is at least $(q-1)^2q^{n-3}$.

\begin{proposition}
    Let $1 < k < n / 2$.
    The maximum number of words in a simultaneously $(1,k)$-overlap-free and $(n-k,n-1)$-overlap-free code is at least
    \begin{align*}
        \biggl \lfloor \frac{qk}{k+2} \biggr \rfloor^{k} \left(q - \biggl \lfloor \frac{qk}{k+2} \biggr \rfloor\right)^2 q^{n-k-2}.
    \end{align*}
\end{proposition}

\begin{proof}
Let $x$ be an integer between 0 and $q$.
Take a partition of $\Sigma$ so that $L_1$ has $x$ elements and $R_1$ has $q-x$ elements. Use Construction~\ref{construction:fimmel} to obtain a non-overlapping code $X$ by setting $L_i$ empty for $i > 1$. Note that $R_i = L_1^{i-1}R_1$ and $X = L_1^k R_1$. 
Now use Construction~\ref{construction:simmultaneous} to get a simultaneously $(1,k)$-overlap-free and $(n-k,n-1)$-overlap-free code
\begin{align*}
    C = \bigcup_{\alpha} X \Sigma^{n-2k-1} R_{\alpha_1} \cdots R_{\alpha_{par(\alpha)}}.
\end{align*}
Its size equals
\begin{align*}
    \lvert C \rvert &= \lvert X \rvert q^{n-2k-1} \sum_{\alpha}  \lvert L \rvert^{k-par(\alpha)} \lvert R \rvert^{par(\alpha)} \nonumber\\
    &= x^k (q-x) q^{n-2k -1} \sum_{l=1}^{k} \binom{k-1}{l-1} x^{k-l}(q-x)^l.
\end{align*}
Applying the binomial theorem we obtain
\begin{align*}
    \lvert C \rvert &= x^{k} (q-x)^2 q^{n-2k -1} (x + q - x)^{k-1} \\
    &= x^{k} (q-x)^2 q^{n-k -2}.
\end{align*}
Since $f$ is differentiable, we can determine from its derivative
    \begin{align*}
        \lvert C \rvert ' (x) = 
        q^{n-k-2}x^k (qk-(k+2)x)(q-x)
    \end{align*}
    that the maximum on $(0,q)$ is at $x = \frac{qk}{k+2}$. 
    Since there are no other local extrema, $f$ restricted to integers attains its maximum at one of the roundings of $x$ to an integer. We again decide to round the value down because for $k > 1$, $1 \leq \bigl \lfloor \frac{qk}{k+2} \bigr \rfloor \leq q - 1$, while whenever $2q < k + 2$, $\bigl \lceil \frac{qk}{k+2} \bigr\rceil > q -1$.
\end{proof}

\section*{Declarations}
\subsection*{Funding}
This research was partially supported by the scientific research program P2-0359 financed by the Slovenian Research and Innovation Agency.

\subsection*{Competing Interests}
The author certifies that they have no affiliations with or involvement in any organisation or entity with any financial or non-financial interest in the subject matter or materials discussed in this manuscript.

\subsection*{Data availability}
Not applicable.


\begin{thebibliography}{10}

\bibitem{cai:2023}
Q.~Cai, X.~Wang, and T.~Feng, ``Constructions and bounds for $q$-ary $(1, k)$-overlap-free codes,'' {\em IEEE Transactions on Information Theory}, 2023.

\bibitem{blackburn:2023}
S.~R. Blackburn, N.~N. Esfahani, D.~L. Kreher, and D.~R. Stinson, ``Constructions and bounds for codes with restricted overlaps,'' {\em IEEE Transactions on Information Theory}, 2023.

\bibitem{levenshtein:1970}
V.~I. Levenshtein, ``Maximum number of words in codes without overlaps,'' {\em Problemy Peredachi Informatsii}, vol.~6, no.~4, pp.~88--90, 1970.

\bibitem{ahlswede:2005}
R.~Ahlswede, B.~Balkenhol, C.~Deppe, H.~Mashurian, and T.~Partner, ``T-shift synchronization codes,'' {\em Electronic Notes in Discrete Mathematics}, vol.~21, pp.~119--123, 2005.

\bibitem{ahlswede:2008}
R.~Ahlswede, B.~Balkenhol, C.~Deppe, H.~Mashurian, and T.~Partner, ``T-shift synchronization codes,'' {\em Discrete applied mathematics}, vol.~156, no.~9, pp.~1461--1468, 2008.

\bibitem{Yazdi:2015}
S.~Tabatabaei~Yazdi, Y.~Yuan, J.~Ma, H.~Zhao, and O.~Milenkovic, ``{A rewritable, random-access DNA-based storage system},'' {\em Scientific reports}, vol.~5, no.~1, pp.~1--10, 2015.

\bibitem{Yazdi:2018}
S.~T. Yazdi, H.~M. Kiah, R.~Gabrys, and O.~Milenkovic, ``{Mutually uncorrelated primers for DNA-based data storage},'' {\em IEEE Transactions on Information Theory}, vol.~64, no.~9, pp.~6283--6296, 2018.

\bibitem{Kumar:2023}
N.~Kumar, S.~S. Bhoi, and A.~K. Singh, ``{A study of primer design with $w$-constacyclic shift over $\mathbb{F}_4$},'' {\em Theoretical Computer Science}, vol.~960, p.~113925, 2023.

\bibitem{Lu:2023}
X.~Lu and S.~Kim, ``{Weakly mutually uncorrelated codes with maximum run length constraint for DNA storage},'' {\em Computers in Biology and Medicine}, vol.~165, p.~107439, 2023.

\bibitem{Chee:2020}
Y.~M. Chee, H.~M. Kiah, and H.~Wei, ``Efficient and explicit balanced primer codes,'' {\em IEEE Transactions on Information Theory}, vol.~66, no.~9, pp.~5344--5357, 2020.

\bibitem{Blackburn:2015}
S.~R. Blackburn, ``Non-overlapping codes,'' {\em IEEE Transactions on Information Theory}, vol.~61, no.~9, pp.~4890--4894, 2015.

\bibitem{petersen:1996}
H.~Petersen, ``On the language of primitive words,'' {\em Theoretical Computer Science}, vol.~161, no.~1-2, pp.~141--156, 1996.

\bibitem{fimmel:2023}
E.~Fimmel, C.~Michel, F.~Pirot, J.-S. Sereni, and L.~Str{\"u}ngmann, ``Diletter and triletter comma-free codes over finite alphabets,'' {\em The Australasian Journal of Combinatorics}, vol.~86, pp.~233--270, 2023.

\bibitem{stanovnik:2024}
L.~Stanovnik, M.~Mo{\v{s}}kon, and M.~Mraz, ``In search of maximum non-overlapping codes,'' {\em Designs, Codes and Cryptography}, pp.~1--28, 2024.

\bibitem{fimmel:2019}
E.~Fimmel, C.~Michel, F.~Pirot, J.-S. Sereni, and L.~Str{\"u}ngmann, ``{Comma-free Codes Over Finite Alphabets}.'' working paper or preprint, Nov. 2019.

\end{thebibliography}
\end{document}